\documentclass{jfm}
\usepackage{graphicx}
\usepackage{epstopdf, epsfig}
\usepackage{comment}
\usepackage{natbib}
\usepackage{listings}
\usepackage{amsmath, amssymb, amsfonts}
\usepackage{appendix}
\usepackage{fancyhdr}
\usepackage{physics}
\usepackage{upgreek}
\usepackage{multicol,multirow}
\usepackage{bbm}
\usepackage[shortlabels]{enumitem}
\usepackage{siunitx}
\usepackage{mathtools}
\usepackage{mathrsfs}
\usepackage{subcaption}
\usepackage{listings}
\usepackage{resizegather}
\usepackage[bottom]{footmisc}
\usepackage{placeins}
\usepackage[colorlinks,allcolors=blue]{hyperref}
\usepackage{ifpdf}
\usepackage[T1]{fontenc}
\usepackage{newtxtext}
\usepackage{newtxmath}
\usepackage{textcomp}
\usepackage{xcolor}
\usepackage{anyfontsize}

\newtheorem{theorem}{Theorem}[section]

\newtheorem{proposition}[theorem]{Proposition}

\newtheorem{definition}[theorem]{Definition}

\renewcommand{\epsilon}{\varepsilon}

\DeclareMathOperator{\One}{\mathbbm{1}}

\renewcommand{\vec}{\mathbf}

\renewcommand{\norm}[1]{\lVert {#1} \rVert}

\shorttitle{Hydrodynamic Control of Microparticles}
\shortauthor{H. Shum, M. Zoppello, M. Astwood, and M. Morandotti}

\title{Control of Microparticles Through Hydrodynamic Interactions}

\author{Henry Shum\aff{1}
    \corresp{\email{henry.shum@uwaterloo.ca}},
 Marta Zoppello\aff{2},
 Michael Astwood\aff{3},
 \and Marco Morandotti\aff{2}}

\affiliation{\aff{1}Department of Applied Mathematics, University of Waterloo,
Waterloo, Ontario, 
Canada N2L 3G1
\aff{2}Department of Mathematical Sciences ``G.~L.~Lagrange'', Politecnico di Torino, Corso Duca degli Abruzzi, 24, 10129, Torino, Italy
\aff{3}Department of Physics and Astronomy, University of Waterloo,
Waterloo, Ontario, 
Canada N2L 3G1}

\begin{document}

\maketitle

\abstract{The controllability of passive microparticles that are advected with the fluid flow generated by an actively controlled one is studied. The particles are assumed to be suspended in a viscous fluid and well separated so that the far-field Stokes flow solutions may be used to describe their interactions. Applying concepts from geometric control theory, explicit moves characterized by a small amplitude parameter~$\varepsilon$ are devised to prove that the active particle can control one or two passive particles. The leading-order (in~$\varepsilon$) theoretical predictions of the particle displacements are compared with those obtained numerically and it is found that the discrepancy is small even when $\varepsilon\approx 1$. These results demonstrate the potential for a single actuated particle to perform complex micromanipulations of passive particles in a suspension. 
}

\section{Introduction}

Manipulation of microparticles suspended in fluids has relevance to several applications, including targetted drug delivery~\citep{nelson_microrobots_2010, li_micro/nanorobots_2017,ezike_advances_2023}, environmental remediation~\citep{wang_catalytic_2016}, cell sorting~\citep{bhagat_microfluidics_2010, wang_size-sensitive_2011}, assisted fertilization~\citep{fishel_micromanipulation_1993}, and microassembly~\citep{ghadiri_microassembly_2012, agnus_robotic_2013}.  

Some common mechanisms for transporting large collections of particles in microfluidic devices are using pressure-driven fluid flow along channels, electrokinetic effects, and acoustic streaming~\citep{chakraborty_microfluidic_2010, wu_acoustofluidic_2019}. It has also been shown that spatially and temporally patterned fluid flows can be generated in microfluidic chambers through buoyancy and electrokinetic effects associated with chemical reactions~\citep{sengupta_self-powered_2014, ortiz-rivera_convective_2016, niu_microfluidic_2017, shum_flow-driven_2018} or by harnessing bacterial or artificial cilia carpets~\citep{darnton_moving_2004, Kim2015}. The motion of individual particles subject to these effects could be dependent on the size or other properties of the particle so there is some control over at least the direction and speed of transport, but these mechanisms are not typically used for fine manipulation of individual particles along specific, arbitrary paths. Instead, techniques such as optical tweezers~\citep{ghadiri_microassembly_2012, bradac_nanoscale_2018}, micropipettes~\citep{zhang_multi-functional_2024}, and externally applied magnetic fields can be used for precise manipulation of individual particles~\citep{khalil_wireless_2012}. 

Optical tweezers are particularly useful and popular for holding and moving cells and other microparticles. This method uses focused laser beams that exert optical forces on particles, preventing a particle from deviating from the center of a trap~\citep{polimeno_optical_2018, bunea_strategies_2019, jamil_optical_2022}. Multiple beams can be formed to trap and move multiple particles simultaneously but this requires more complicated experimental procedures, making it inconvenient as a method for controlling many particles. Moreover, 
optical tweezers face challenges and limitations from heating effects, the dependence on the size and dielectric properties of the particle being trapped, and degraded laser focus when particles are located deeper in the fluid, far from the bounding glass surface~\citep{melzer_fundamental_2018, jamil_optical_2022, malinowska_introduction_2024}.

In the current work, we explore a strategy for manipulating passive particles suspended in a viscous fluid relying on hydrodynamic interactions with a set of active particles directly controlled by other means.
For example, the active particles could be controlled by optical tweezers or externally imposed magnetic forces~\citep{khalil_wireless_2012}.
This would allow us to use a minimal controlling apparatus to manipulate passive particles.
Exploiting the hydrodynamic interaction has several advantages: forces propagate throughout the whole fluid and influence particles that are far away from one that is moving (or being moved), overcoming both the spatial limitation of optical tweezers and the need to have one tweezer for each particle to be displaced.
Another advantage is that, in the far-field approximation, the size of the passive particles does not enter the equations of motion, so, in principle, one can control arbitrarily large particles (provided other modelling assumptions remain valid).

The general setup for our model is that a collection of neutrally buoyant particles is suspended in an incompressible Newtonian fluid. Since we are primarily interested in microfluidic systems, it is natural to adopt the equations of incompressible Stokes flow to characterize the behavior of the fluid in low Reynolds number applications. For particles separated by a distance $r=100$\,\textmu{}m and moving with characteristic speeds $U=100$\,\textmu{}m\,s$^{-1}$ in a fluid with the kinematic viscosity $\nu = 1$\,mm$^2$\,s$^{-1}$ (comparable to that of water at 20$^\circ$C), the Reynolds number is $\Rey =Ur/\nu=0.01$, for example. 

We consider the case where one of these particles is directly controlled by external forces, so that its velocity and position are prescribed functions of time, and the remaining particles move passively in the fluid flow generated by the actively controlled one. We use the far-field expressions for the Stokes flow field produced by a translating rigid, spherical particle to determine the velocity and trajectory of passive particles in the fluid. 

Using explicit constructions inspired by the work of~\cite{dalmaso}, we prove total controllability of systems consisting of one active and either one or two passive particles. That is, such particles can be moved from arbitrary initial positions to arbitrary final positions in unbounded three-dimensional space, provided that the particles are far apart from each other in these configurations so that the far-field approximation is valid.

Controllability of a single passive particle by one active particle was proved by~\cite{GaffneyMoreau} using abstract tools from geometric control theory (see, e.g.,~\cite{Agrachev2004}). In particular, they showed that the Lie brackets generated by the vector fields controlling the velocity of the active particle spanned the full six-dimensional configuration space for one active and one passive particle: controllability follows owing to the Rashewsky--Chow Theorem~\cite[Theorem~5.9]{Agrachev2004}.

We provide, in contrast, a strategy that breaks down the desired displacements into a sequence of steps that can be achieved by iteratively applying elementary moves, in which a passive particle moves in the radial direction or in the polar direction around the active particle, for example. We extend the controllability result to two passive particles, proposing a  separate strategy for this case. Using numerical solutions, we test the accuracy of asymptotic expressions for displacements predicted for our elementary moves and also assess the errors associated with adopting the far-field hydrodynamic approximation. 

Our results are a step toward the more general problem of independently manipulating an arbitrary number of passive particles using a small number of control variables. Although the strategies we discuss in the current work do not readily extend to larger numbers of particles, the framework can be generalized to describe such systems in a straight-forward manner.

\smallskip

The paper is organized as follows: in Section~\ref{sec:math_form}, we present the mathematical formulation of the problem; in Section~\ref{sec:moves}, we construct the elementary and compound moves that will be used in Section~\ref{sec:controllability_M_1_2} to prove controllability of our system of one active particle and one or two passive ones.
In Section~\ref{sec:errors}, we discuss the errors due to finite amplitudes and separations in comparison to the theoretically predicted trajectories.
Finally, in Section~\ref{sec:conclusions}, we offer an overview of the results we obtained and an outlook for future research.

\section{Mathematical formulation}\label{sec:math_form}
We introduce the general setup for a system of spherical particles immersed in a viscous fluid. Of these, one is an \emph{active particle}, whose velocity is directly prescribed, and the remaining~$M$ are \emph{passive particles}, whose motions are determined by the interaction with the active one.

We let $\rho_0>0$ be the radius of the active particle and we denote by $t\mapsto\vec{x}(t)$ its position in space at time $t$ and by $t\mapsto\dot{\vec{x}}(t)$ its velocity at time $t$. 
Analogously, we denote by $\rho_j>0$, $t\mapsto\vec{y}_j(t)$, and $t\mapsto\dot{\vec{y}}_j(t)$, for every $j=1,\ldots,M$, the radius, position, and velocity, respectively, of the $j${th} passive particle. 

Assuming that the Reynolds number is small enough that inertial effects may be neglected, the fluid flow is governed by the equations of incompressible Stokes flow, 
\begin{equation}\label{eq:stokes}
    \nabla p -\mu \nabla^2\vec{u}=\vec{0}, \qquad \nabla\cdot \vec{u} = 0,
\end{equation}
where $\vec{u}$ is the velocity field, $p$ is the pressure field, and $\mu$ is the dynamic viscosity of the fluid. We assume that the velocity field vanishes at infinity and satisfies the no-slip boundary conditions on the surfaces of the particles, namely,
\begin{equation*}
    \vec{u}(\vec{z}) = \begin{cases}
        \vec{U}_0 + \vec{\Omega}_0\times (\vec{z} - \vec{x}), & \text{for } \lVert\vec{z}-\vec{x}\rVert = \rho_0,\\
        \vec{U}_j + \vec{\Omega}_j\times (\vec{z} - \vec{y}_j), & \text{for } \lVert\vec{z}-\vec{y}_j\rVert = \rho_j,\; j= 1,\ldots,M,
   \end{cases} 
\end{equation*}
where $\vec{U}_j$ and $\vec{\Omega}_j$ ($j=0,\ldots,M$) are the translational and rotational velocities, respectively, of the active and passive particles. % $\alpha$ from population $\star \in \{A,P\}$. 
We denote by $t\mapsto\vec{F}(t)$ the force that the active particle exerts, at time $t$, on the surrounding fluid and we assume that all passive particles are force-free. All particles, whether active or passive, are torque-free.

By linearity of the equations of Stokes flow, the relationship between active forces and the velocities of the particles, in the absence of background flows, are generically described by~\citep{Kim1991}
\begin{equation*}%\label{eq:mobilityequation}
%    \begin{pmatrix}
%        \vec{U}_0 \\ \vec{\Omega}_0
%    \end{pmatrix} = 
%    \begin{pmatrix}
%        \vec{M}_0 \\ \vec{N}_0
%    \end{pmatrix} \vec{F},\qquad
    \begin{pmatrix}
        \vec{U}_j \\ \vec{\Omega}_j
    \end{pmatrix} = 
    \begin{pmatrix}
        \vec{M}_{j} \\ \vec{N}_{j}
    \end{pmatrix} \vec{F},\quad\text{for $j=0,\ldots,M$.}
\end{equation*}
The quantities $\vec{M}_{j}$ ($j=0,\ldots,M$) are the mobility tensors for the translational velocities of the active and passive particles, respectively, due to the force on the active particle, and $\vec{N}_{j}$ ($j=0,\ldots,M$) are the mobility tensors for the rotational velocities of the active and passive particles, respectively, due to the force on the active particle. 
In general, the mobility tensors depend on the relative positions of all particles and it is not possible to obtain a closed-form expression for them. 
By symmetry of the spherical particles, the mobility tensors are independent of the orientations of the particles. We focus on the problem of controlling the particle positions without regard to their orientations, hence, the rotational velocities need not be considered.

Let $\vec{d}_{j}\coloneqq \vec{y}_j-\vec{x}$ ($j=1,\ldots,M$) be the displacement vector of the $j${th} passive particle from the active one.
Assuming that all of the pairwise distances $r_{j}\coloneqq \lVert\vec{d}_{j}\rVert$ %and $r_{j'j}\coloneqq \lVert\vec{d}_{j'j}\rVert$ 
are large compared with all of the particle radii $\rho_0$ and $\rho_j$, as well as the mutual distances between the passive particles, we use the far-field approximation for the translational mobility tensors~\citep{zuk_rotneprageryamakawa_2014, Graham2018} given by
\begin{equation}\label{MM02}
\vec{M}_{j}=\begin{cases}
    \displaystyle \frac{1}{6\pi\mu \rho_0}\One, & \text{for $j=0$},\\[2ex] 
    \displaystyle \bigg(1+\frac{\rho_0^2+\rho_j^2}{6}\nabla^2\bigg)\mathcal{G}(\vec{d}_j), & \text{for $j=1,\ldots,M$,}
\end{cases}
\end{equation}
where the function $\mathcal{G}\colon\mathbb{R}^3\setminus\{\vec{0}\}\to\mathbb{R}^{3\times3}_{\mathrm{sym}}$ defined by
\begin{equation*}%\label{MM03}
\mathcal{G}(\vec{d}) \coloneqq \frac{1}{8\pi \mu}\left(\frac{\One}{r}+\frac{\vec{d}\otimes \vec{d}}{r^3}\right)
\end{equation*}
is the Stokeslet Green's function for the Stokes equation~\eqref{eq:stokes} with a singular force $\vec{F}\delta(\vec{d})$ applied at the origin. Equation \eqref{MM02} is accurate up to $O(r_j^{-3})$ and can be extended to higher orders of accuracy by the method of reflections~\citep{Kim1991}.
Note that, to this order of accuracy, the passive particles do not affect the velocities of the active particles. 

The matrix $\vec{M}_0$ is evidently invertible and its inverse $\vec{R}_0=\vec{M}_0^{-1}=6\pi\mu\rho_0\One$ is the resistance matrix describing the linear relationship between forces applied to the fluid and the velocities of the particles, $\vec{F} = \vec{R}_0\vec{U}_0$. 
The equations of motion for our system of active and passive particles are then
\begin{equation}\label{MM01}
\begin{cases}
\dot{\vec{x}}(t)=\vec{u}(t), & %\text{for $i=1,\ldots,N$,} 
\\[2mm]
\displaystyle \dot{\vec{y}}_j(t)= \vec{M}_{j}(t)\vec{F}(t) =  \overline{\vec{M}}_{j}(t)\vec{u}(t),
& \text{for $j=1,\ldots,M$,}
\end{cases}
\end{equation} 
where $\overline{\vec{M}}_{j}(t) = \vec{M}_{j}(t) \vec{R}_0(t)=6\pi\mu\rho_0\vec{M}_j(t)$. 
We further simplify the equations by retaining only the leading order terms, namely, those of order $O(r_j^{-1})$. Notice that, with this approximation, the radii of the passive particles do not enter the system of equations, which becomes
\begin{equation}\label{MM05}
\begin{cases}
\dot{\vec{x}}(t)=\vec{u}(t), \\[2mm]
\displaystyle \dot{\vec{y}}_j(t)=\frac{3a}{4}\bigg(\frac{\One}{r_{j}(t)}+\frac{\vec{d}_{j}(t)\otimes \vec{d}_{j}(t)}{r_{j}(t)^3}\bigg)\vec{u}(t) = \vec{G}(\vec{d}_j(t))\vec{u}(t),
& \text{for $j=1,\ldots,M$,}
\end{cases}
\end{equation}
where $\vec{G}(\vec{d}) \coloneqq \frac{3a}{4}\left(\frac{\One}{r}+\frac{\vec{d}\otimes \vec{d}}{r^3}\right)$ and $a = \rho_0$. We remark that the passive particles move as tracers or point particles in the flow field induced by the moving active particles.
Equation \eqref{MM05} can be written in matrix form as 
\begin{equation}\label{MM07}
\begin{pmatrix}
\dot{\vec{x}} \\ \dot{\vec{y}}_1 \\ \vdots \\ \dot{\vec{y}}_M
\end{pmatrix}=\vec{H}\vec{u},\quad \text{where}\quad \vec{H}=\begin{pmatrix}
\One_{3\times3} \\
\vec{G}_{1} \\
\vdots \\
\vec{G}_{M}\end{pmatrix}, \quad \vec{G}_{j} = \vec{G}(\vec{d}_j), \text{ for } j=1,\ldots,M.
\end{equation}

\section{Elementary and compounds moves}\label{sec:moves}
In this section, we construct the elementary and compound moves, which are the building blocks for proving the controllability results in Section~\ref{sec:controllability_M_1_2}.

\subsection{Elementary moves and Lie brackets of vector fields}
\label{subsec:Lie_brackets}
We describe three elementary classes of control functions from which we will construct strategies to move the active and passive particles from arbitrary initial positions to arbitrary target positions. 
Since the equations for the passive particles are decoupled, \emph{i.e.}, the velocity of the $j$th passive particle does not depend on the $i$th passive particle with $i\neq j$, the action of the active particle is the same on all the passive particles.
For this reason, in what follows, the three elementary classes of control functions will be described for the case $M=1$ in equations~\eqref{MM05} and~\eqref{MM07}.
Let $\vec{h}_k\in\mathbb{R}^{6}$ be the $k$-th column of $\vec{H}$ in~\eqref{MM07}, for $M=1$ and $k=1,2,3$. The first three components of $\vec{h}_k$ represent the velocity of the active particle and the last three components represent the velocity of the passive particle when the control $\vec{u}=\vec{e}_k$ is applied. 

\subsubsection{The zeroth-order control\label{sec:direct_control}}
Consider the zeroth-order (constant) control 
\begin{equation*}%\label{direct_control}
    \vec{u}(t) = \vec{u}_{\Delta t}^{\epsilon,\,\vec{h}_k}(t) \coloneqq\begin{cases}
        \displaystyle \frac{\epsilon}{\Delta t} \vec{e}_k\,, & 0\leq t < \Delta t,\\[2mm]
        0,& \text{otherwise.}        
    \end{cases}
\end{equation*}
The net displacements of the active and passive particles over the time interval $[0,\Delta t]$ with initial positions $\vec{x}(0)=\vec{x}^\circ$ and $\vec{y}(0)=\vec{x}^\circ+\vec{d}$ (with $\vec{d}\in\mathbb{R}^3\setminus\{\vec{0}\}$), respectively, are
\begin{equation}\label{formula11}
\begin{split}
    \Delta^{\epsilon,\,\vec{h}_k}(\vec{d}) = \begin{pmatrix}
        \Delta_\vec{x}^{\epsilon,\,\vec{h}_k}(\vec{d})\\ \Delta_\vec{y}^{\epsilon,\,\vec{h}_k}(\vec{d})
    \end{pmatrix}\coloneqq &
    \begin{pmatrix}
    \vec{x}(\Delta t)-\vec{x}(0)\\
        \vec{y}(\Delta t)-\vec{y}(0)
    \end{pmatrix} = \begin{pmatrix}
        \epsilon\, \vec{e}_k \\
        \frac{3a\epsilon}{4r}\left(\vec{e}_k+\frac{d_k \vec{d}}{r^2}\right) + O(\epsilon^2)
    \end{pmatrix} \\
    = &  \,
    \epsilon \vec{h}_k(\vec{d}) + \begin{pmatrix}
        \vec{0}\\ O(\epsilon^2)
    \end{pmatrix}
\end{split}
\end{equation}
as $\epsilon \to 0$.
This zeroth-order control is primarily used to control the position of the active particle since its velocity is directly prescribed by the control. The passive particle will also move, due to the flow field generated by the translating active particle and we can compute the trajectory of the passive particle according to \eqref{formula11}, see Figure~\ref{fig:vectorfields}(a).

To displace the passive particle without any net displacement of the active particle, we construct higher-order controls.
\begin{figure}
\centering
    \caption{Streamline plots of the displacement vector field of a passive particle, given by components 4--6 of the Lie brackets, for (a) the $\vec{h}_1$ zeroth-order control,
    (b) the $[\vec{h}_1,\vec{h}_2]$ first-order control,
    (c) the $x$--$y$ plane of the $[\vec{h}_2,[\vec{h}_1,\vec{h}_2]]$ second-order control, and
    (d) the $x$--$z$ plane of the $[\vec{h}_2,[\vec{h}_1,\vec{h}_2]]$ second-order control. In all cases, we use an active particle of radius $a=1$ and the plots are in the same plane as the active particle, which is located at the origin. By symmetry, there are no out-of-plane components of displacements for the passive particle
    The color scale indicates the base 10 logarithm of the magnitude.}
    \label{fig:vectorfields}
\end{figure}

\subsubsection{The first-order control\label{sec:rotation_control}}
Consider a control that moves the active particle around a square loop with sides of length $\epsilon$ in the $\vec{e}_k$ and $\vec{e}_l$ directions, namely,
\begin{equation}
 \vec{u}_{\Delta t}^{\epsilon,\, [\vec{h}_k,\vec{h}_l]}(t) \coloneqq
    \vec{u}_{\Delta t/4}^{\epsilon,\, \vec{h}_k}(t) + 
    \vec{u}_{\Delta t/4}^{\epsilon,\, \vec{h}_l}(t-\tfrac{\Delta t}4) -
    \vec{u}_{\Delta t/4}^{\epsilon,\, \vec{h}_k}(\tfrac{3\Delta t}4-t) -
    \vec{u}_{\Delta t/4}^{\epsilon,\, \vec{h}_l}(\Delta t - t),
\label{eq:squarecontrol}
\end{equation}
for $k,l=1,2,3$. The inversion of the time variable in the last two terms represents performing the reverse control of the first two terms in the sum. The controls here are constant over each subinterval but in later examples, it will be necessary to make this distinction. Explicitly, the function in \eqref{eq:squarecontrol} can be expressed as
\begin{equation*}
\vec{u}_{\Delta t}^{\epsilon,\, [\vec{h}_k,\vec{h}_l]}(t) =
\begin{cases}
        \displaystyle\frac{4\epsilon}{\Delta t}\vec{e}_k, & \displaystyle0\leq t < \frac{\Delta t}4,\\[2mm]
        \displaystyle\frac{4\epsilon}{\Delta t}\vec{e}_l, & \displaystyle\frac{\Delta t}4 \leq t < \frac{\Delta t}2,\\[2mm]
        \displaystyle-\frac{4\epsilon}{\Delta t}\vec{e}_k, & \displaystyle\frac{\Delta t}2 \leq t < \frac{3\Delta t}4,\\[2mm]
        \displaystyle-\frac{4\epsilon}{\Delta t}\vec{e}_l, & \displaystyle\frac{3\Delta t}4 \leq t < {\Delta t},\\[2mm]
        0, & \text{otherwise.}        
    \end{cases}%\label{rotation_control}
    \end{equation*}
The net displacements, for small $\epsilon$, are 
\begin{equation}
    \Delta^{\epsilon,\,[\vec{h}_k,\vec{h}_l]}(\vec{d}) = \begin{pmatrix}
        \Delta_\vec{x}^{\epsilon,\,[\vec{h}_k,\vec{h}_l]}(\vec{d})\\ \Delta_\vec{y}^{\epsilon,\,[\vec{h}_k,\vec{h}_l]}(\vec{d})
    \end{pmatrix} = 
    \epsilon^2 [\vec{h}_k, \vec{h}_l](\vec{d}) + \begin{pmatrix}
        \vec{0}\\ O(\epsilon^3)
    \end{pmatrix},
    \label{eq:first-order_displacement}
\end{equation}
where $[\vec{h}_k,\vec{h}_l] \coloneq \vec{h}_k\cdot\nabla\vec{h}_l - \vec{h}_l\cdot\nabla\vec{h}_k$ is the first-order Lie bracket of the vector fields $\vec{h}_k$ and $\vec{h}_l$. The Lie bracket evaluates to
\begin{equation*}%\label{rotation_bracket}
    [\vec{h}_k,\vec{h}_l](\vec{d}) = \begin{pmatrix}
        \vec{0}\\
        \frac{3a}{2r^3}\left(1-\frac{9a}{8r}\right)(d_k\vec{e}_l - d_l\vec{e}_k)
    \end{pmatrix} = 
    \begin{pmatrix}
        \vec{0}\\
        \pmb{\upomega}\times\vec{d}
    \end{pmatrix},
\end{equation*}
where $\pmb{\upomega}(r) = \frac{3a}{2r^3}\left(1-\frac{9a}{8r}\right)\vec{e}_k\times\vec{e}_l$ and $\vec{0}$ is the three-dimensional zero vector. Hence, for $k\neq l$ and small $\epsilon$, the control $\vec{u}_{t_0,\, \Delta t}^{\epsilon,\,[\vec{h}_k,\vec{h}_l]}$ results approximately in a rotation of the passive particle by an angle 
\begin{equation}\label{Deltatheta}
\Delta\theta^\varepsilon(r) = \frac{3a\epsilon^2}{2r^3}\left(1-\frac{9a}{8r}\right)
\end{equation}
about the axis passing through the active particle and perpendicular to $\vec{e}_k$ and $\vec{e}_l$. By construction, the active particle returns to its initial position, $\vec{x}(\Delta t)=\vec{x}(0)$. If $k=l$, then the net displacements are exactly zero as this is a time-reciprocal motion. 

The interpretation of this result is that a particle forced to move around in a closed, square loop produces a net displacement field that is, to leading order, equivalent to a rotlet, see Figure~\ref{fig:vectorfields}(b). Indeed, the time-averaged distribution of forces applied to the fluid over the interval $[0,\Delta t]$ corresponds to the sum of a Stokeslet dipole with force in the $\vec{e}_l$ direction and displacement in the $\vec{e}_k$ direction and a Stokeslet dipole with force in the $-\vec{e}_k$ direction and displacement in the $\vec{e}_l$ direction. 

We refer to the control $\vec{u}_{\Delta t}^{\epsilon,\, [\vec{h}_k,\vec{h}_l]}$ as the first-order control (and assume that $k\neq l$) since it corresponds to a first-order Lie bracket. 

Since rotations preserve the distance $r$, we require another class of controls: one that generates net displacements of the passive particle in the radial direction, with respect to the active particle, without a net displacement of the active particle.

\subsubsection{The second-order control}\label{sec:translation_control}
Consider second-order control functions of the form $\vec{u}(t) = \vec{u}_{\Delta t}^{\epsilon,\, [\vec{h}_k,[\vec{h}_l, \vec{h}_m]]}(t)$ defined as in~\eqref{eq:squarecontrol}, replacing $\vec{h}_l$ with $[\vec{h}_l,\vec{h}_m]$. Since $[\vec{h}_l, \vec{h}_m]$ corresponds to a rotlet-like flow field if $l\neq m$, the Lie bracket $[\vec{h}_k,[\vec{h}_l, \vec{h}_m]]$ has the approximate form of a rotlet dipole, with axis $(\vec{e}_l\times \vec{e}_m)$ and displacement in the $\vec{e}_k$ direction, acting on the passive particle,
\begin{equation*}%\label{translation_bracketformula}
\begin{split}
    [\vec{h}_k,[\vec{h}_l,\vec{h}_m]](\vec{d}) &= \begin{pmatrix}
        \vec{0}\\
        \frac{3a}{2r^3}\left(1-\frac{9a}{8r}\right)(\delta_{mk}\vec{e}_l - \delta_{lk}\vec{e}_m) + \frac{9a}{2r^3}\left(1-\frac{3a}{2r}\right)^2\frac{d_k}{r}\left(\frac{d_l\vec{e}_m - d_m\vec{e}_l}{r}\right) - \frac{9a^2}{8r^4}\left(1-\frac{9a}{8r}\right)\delta_{mk}\left(\vec{e}_l+\frac{d_l\vec{d}}{r^2}\right)
    \end{pmatrix}\\
    &= \begin{pmatrix}
        \vec{0}\\
        \frac{3a}{2r^3}\left\{(\vec{e}_k\times(\vec{e}_l\times\vec{e_m}) +\frac{3d_k(\vec{e}_l\times\vec{e}_m)\times\vec{d}}{r^2}\right\} + O(1/r^4) 
    \end{pmatrix},
    \end{split}
\end{equation*}
see Figure~\ref{fig:vectorfields}(c).

The net displacements of the active and passive particles are given, for small $\varepsilon$, by
\begin{equation*}%\label{translation_bracket}
    \Delta^{\epsilon,\,[\vec{h}_k,[\vec{h}_l,\vec{h}_m]]}(\vec{d}) = \begin{pmatrix}
        \Delta_\vec{x}^{\epsilon,\,[\vec{h}_k,[\vec{h}_l,\vec{h}_m]]}(\vec{d})\\ \Delta_\vec{y}^{\epsilon,\,[\vec{h}_k,[\vec{h}_l,\vec{h}_m]]}(\vec{d})
    \end{pmatrix} = 
    \epsilon^3 [\vec{h}_k,[\vec{h}_l,\vec{h}_m]](\vec{d}) + \begin{pmatrix}
        \vec{0}\\ O(\epsilon^4)
    \end{pmatrix}.
\end{equation*}
Notice that if, for a given relative position of the passive particle with respect to the active particle, we choose a right-handed reference frame in which $\vec{d} = d_1\vec{e}_1$, then the control corresponding to $[\vec{h}_2,[\vec{h}_1,\vec{h}_2]]$ results in a passive particle displacement 
\begin{equation}
    \Delta_\vec{y}^{\epsilon,\,[\vec{h}_2,[\vec{h}_1,\vec{h}_2]]}(d_1\vec{e}_1) = \frac{3a\epsilon^3}{2r^3}\left(1-\frac{3a}{2r}\right)\left(1-\frac{9a}{8r}\right)\vec{e}_1 + O(\epsilon^4),
    \label{translation_displacement}
\end{equation}
where $r=|d_1|$. To leading order in $\epsilon$, this produces a displacement in the $\vec{e}_1$ direction.

In the more general configuration assuming only that $d_2 = 0$, we have the result that 
\begin{equation}
    \Delta_\vec{y}^{\epsilon,\,[\vec{h}_2,[\vec{h}_1,\vec{h}_2]]}(d_1\vec{e}_1+d_3\vec{e}_3) = \frac{3a\epsilon^3}{2r^3}\vec{e}_1 + O(1/r^4) + O(\epsilon^4),
    \label{translation_displacement2}
\end{equation}
which implies that, to leading order in $1/r$ and $\epsilon$, the displacement of passive particles in the plane $d_2=0$ is purely in the $\vec{e}_1$ direction and the magnitude of the displacement depends on the magnitude but not the direction of the vector $\vec{d}$. This is illustrated in Figure~\ref{fig:vectorfields}(d).

\subsection{Compound moves}\label{sec:compound_moves}

In this section, we establish a set of manipulations that can be performed on a system of one active and two passive particles (i.e., $M=2$), assuming that their motion is governed by system~\eqref{MM05}, using the elementary moves discussed in Section~\ref{subsec:Lie_brackets}. Since this system is based on the far field approximation, we will ensure that the particles always remain well separated, according to the following definition. 

\begin{definition}\label{def:wellseparated}
Let $R>0$ be the minimum separation we wish to maintain between any two particles. We say that an instantaneous configuration of active and passive particles is \emph{well separated}, and denote this by $(\vec{x}(t),\vec{y}_1(t),\ldots,\vec{y}_M(t)) \in \mathcal{S}_R^{1+M}$, if the minimum distance between any two of the particles at time $t$ is greater than $R$. We say that a solution, or trajectory, $(\vec{x},\vec{y}_1,\ldots,\vec{y}_M)\in (AC([0,T];\mathbb{R}^3))^{1+M}$ of system \eqref{MM05}  is \emph{well separated} if the configuration $(\vec{x}(t),\vec{y}_1(t),\ldots,\vec{y}_M(t)) \in \mathcal{S}_R^{1+M}$ for all times $t\in[0,T]$. 

\end{definition}

For brevity, we will not reiterate the well-separated conditions in Propositions~\ref{prop:make_nonequidistant}--\ref{prop:adjust_distances} that follow, but these conditions will be implied in all cases, namely, we will always move the particles from an initial configuration $(\vec{x}^\circ,\vec{y}_1^\circ,\vec{y}_2^\circ)\in\mathcal{S}_R^3$ to a final configuration $(\vec{x}^f,\vec{y}_1^f,\vec{y}_2^f)\in\mathcal{S}_R^3$ ensuring that the particles stay well separated at all times.

%\subsubsection{}
\begin{proposition}[Equidistant to non-equidistant configurations]\label{prop:make_nonequidistant}
The particles can be moved from any non-collinear initial configuration $(\vec{x}^\circ,\vec{y}_1^\circ,\vec{y}_2^\circ)$ with the two passive particles equidistant from the active particle, i.e., $\norm{\vec{y}_1^\circ - \vec{x}^\circ} = \norm{\vec{y}_2^\circ - \vec{x}^\circ}$, to a final configuration $(\vec{x}^f,\vec{y}_1^f,\vec{y}_2^f)$ satisfying $\norm{\vec{y}_1^f - \vec{x}^f} \neq \norm{\vec{y}_2^f - \vec{x}^f}$.
\end{proposition}
\begin{proof}
    The goal can be achieved using the second-order control.
    We choose a right-handed orthonormal reference frame $\{\vec{e}_1,\vec{e}_2,\vec{e}_3\}$ in which $\vec{e}_3 = \vec{d}_1^\circ/\norm{\vec{d}_1^\circ}$, $\vec{d}_1$ and $\vec{d}_2$ belong to the span of $\vec{e}_1$ and $\vec{e}_3$\, so that ${d}_{i,2}=0$ for $i=1,2$, and $\vec{e}_1\cdot \vec{d}_2 > 0$ [see Figure~\ref{fig:compound}(a)]. Applying the second-order control $\vec{u}_{\Delta t}^{\epsilon, [\vec{h}_2,[\vec{h}_1,\vec{h}_2]]}$ causes no net displacement of the active particle and displaces each passive particle by the same distance in the $\vec{e}_1$ direction (to leading order), according to 
    \eqref{translation_displacement2}. Moving both passive particles in the $\vec{e}_1$ direction breaks the symmetry and results in a final configuration with $\norm{\vec{d}_2^f} > \norm{\vec{d}_1^f}$. To leading order, the distance between passive particles is unchanged and the distance between each passive particle and the active particle is increased so the configuration remains well separated. 
\end{proof}

\begin{figure}
    \centering
    \includegraphics[scale=0.3]{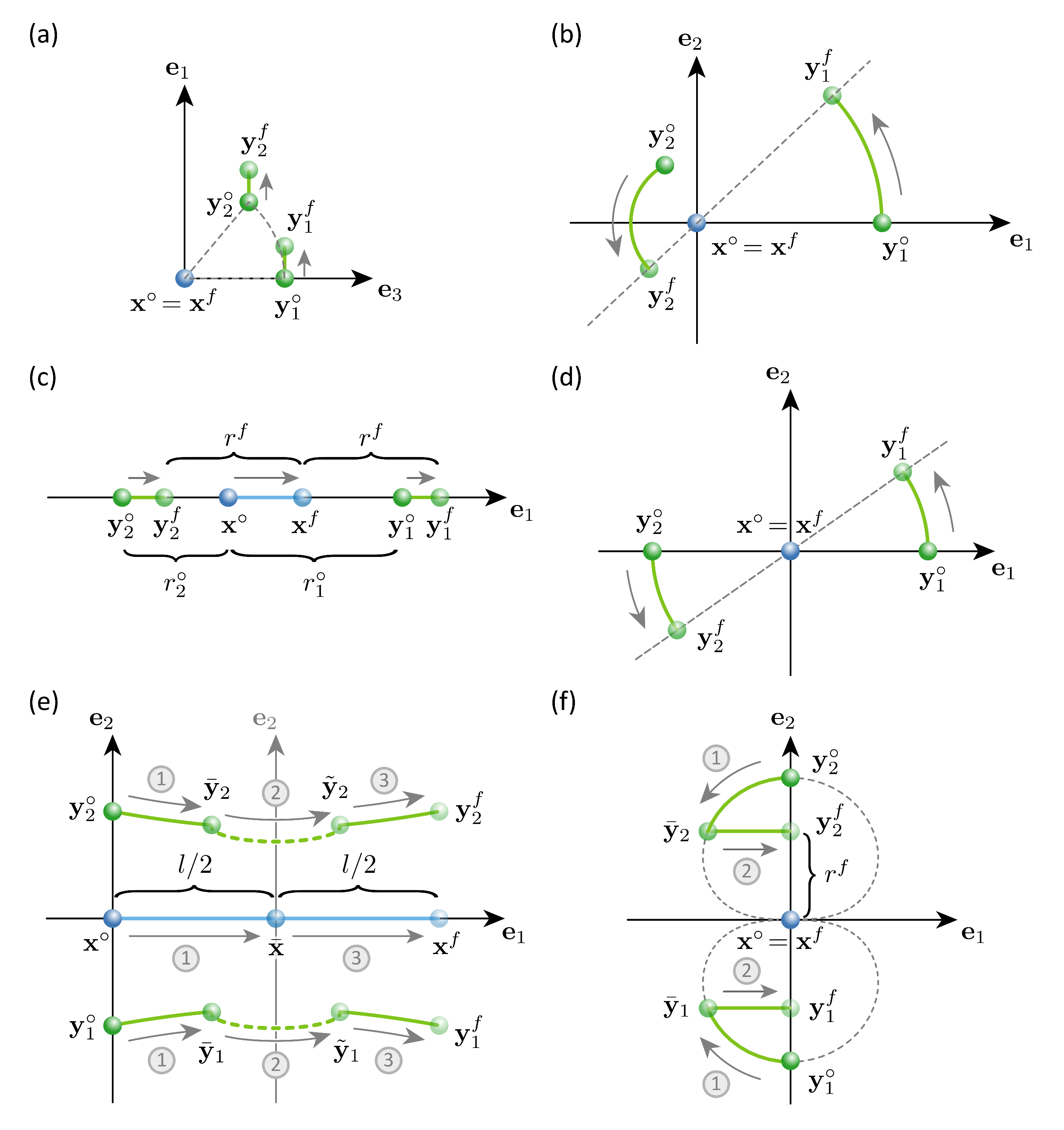}
    \caption{Schematic illustration of the compound moves for two passive particles. The objectives of these moves are to (a) produce a non-equidistant configuration, (b) produce a collinear configuration, (c) make collinear particles equidistant, (d) rotate equidistant collinear particles to a new orientation, (e) translate equidistant collinear particles together, and (f) adjust the distance between equidistant collinear particles. 
    \label{fig:compound}}
\end{figure}

\begin{proposition}[Arbitrary to collinear configurations]\label{prop:make_collinear}
The particles can be moved from any initial configuration $(\vec{x}^\circ,\vec{y}_1^\circ,\vec{y}_2^\circ)$ to a collinear final configuration with the active particle between the two passive particles, i.e., a configuration $(\vec{x}^f,\vec{y}_1^f,\vec{y}_2^f)$ satisfying $\vec{d}_1^f = -\alpha\vec{d}_2^f$ for some $\alpha > 0$.
\end{proposition}
\begin{proof}
    Suppose that the initial configuration does not satisfy the desired final properties. 
    We may also suppose without loss of generality that $\norm{\vec{d}_1^\circ} > \norm{\vec{d}_2^\circ}$. If this is not the case, then we apply Proposition~\ref{prop:make_nonequidistant} to produce a new configuration in which the first passive particle is farther from the active particle. 

    We choose a right-handed orthonormal reference frame $\{\vec{e}_1,\vec{e}_2,\vec{e}_3\}$ in which $\vec{e}_1 = \vec{d}_1/\norm{\vec{d}_1}$, $\vec{d}_1$ and $\vec{d}_2$ belong to the span of $\vec{e}_1$ and $\vec{e}_2$\,, and $\vec{e}_2\cdot\vec{d}_2 \geq 0$ [see Figure~\ref{fig:compound}(b)]. 
    
    Ignoring the $O(\varepsilon^3)$ terms, applying the first-order control $\vec{u}_{\Delta t}^{\epsilon, [\vec{h}_1,\vec{h}_2]}$ with small $\epsilon$ causes no net displacement of the active particle and rotates each of the passive particles around the active particle about the $\vec{e}_3$-axis in the counterclockwise direction by small angles $\Delta\theta_1$ and $\Delta\theta_2$ given by \eqref{Deltatheta}. In the far field, $\Delta \theta$ decreases with $r$, so $\Delta\theta_1 < \Delta\theta_2$ so the angle between $\vec{d}_1$ and $\vec{d}_2$ increases by $(\Delta\theta_2-\Delta\theta_1)$.
    
    Repeating the application of this control with a suitable choice of $\varepsilon$, we can rotate the particles to the final desired collinear configuration with the two passive particles on opposite sides of the active particle. The distances from the active particle to the passive ones are unchanged by repeated applications of the first-order control while the distance between the two passive particles increases so the system remains well separated. 
\end{proof}

\begin{proposition}[Collinear to equidistant collinear configurations]\label{prop:make_equidistant}
The particles can be moved from any collinear initial configuration $(\vec{x}^\circ,\vec{y}_1^\circ,\vec{y}_2^\circ)$ with $\vec{d}_1^\circ = -\alpha \vec{d}_2^\circ$ for some $\alpha > 0$ to a collinear final configuration with the two passive particles equidistant from the active particle, i.e., a configuration $(\vec{x}^f,\vec{y}_1^f,\vec{y}_2^f)$ satisfying $\vec{d}_1^f = -\vec{d}_2^f$.
\end{proposition}
\begin{proof}
    Suppose without loss of generality that $\norm{\vec{d}_1^\circ} > \norm{\vec{d}_2^\circ}$ and that $\vec{e}_1 = \vec{d}_1^\circ/\norm{\vec{d}_1^\circ}$ [see Figure~\ref{fig:compound}(c)]. Applying the zeroth-order control $\vec{u}_{\Delta t}^{\epsilon, \vec{h}_1}$ moves all three particles in the $\vec{e}_1$ direction according to \eqref{formula11}. The active particle translates by the largest magnitude, $\epsilon$, and the second passive particle translates by a larger distance than the first particle because the displacement decreases with $r$. Hence, $\norm{\vec{d}_1}$ decreases and $\norm{\vec{d}_2}$ increases with each application of the zeroth-order control. Note that, within the far field assumption, the active particle can approach the first passive particle arbitrarily closely by repeated applications of the control. Hence, by continuity, a point can be reached at which $\norm{\vec{d}_1^f} = \norm{\vec{d}_2^f}$. The configurations are always well separated if the initial configuration is.
\end{proof}

\begin{proposition}[Reorienting equidistant collinear configurations]\label{prop:rotatingparticles}
The particles can be moved from any equidistant collinear initial configuration $(\vec{x}^\circ,\vec{y}_1^\circ,\vec{y}_2^\circ)$ with $\vec{d}_1^\circ = -\vec{d}_2^\circ$ to any other equidistant collinear final configuration $(\vec{x}^f=\vec{x}^\circ,\vec{y}_1^f,\vec{y}_2^f)$ with the same position of the active particle and the same distance between the active and passive particles, i.e., a configuration with $\vec{d}_1^f = -\vec{d}_2^f$ and $\norm{\vec{d}_1^f} = \norm{\vec{d}_1^\circ}$.
\end{proposition}
\begin{proof}
    We choose a right-handed orthonormal reference frame $\{\vec{e}_1,\vec{e}_2,\vec{e}_3\}$ in which $\vec{d}_1^\circ$ and $\vec{d}_1^f$ lie in the span of $\vec{e}_1$ and $\vec{e}_2$. Then, the transformation from the initial to the desired final configuration can be described as a rotation about the $\vec{e}_3$ axis through the active particle by an angle $\theta = \arccos\left(\vec{d}_1^\circ\cdot \vec{d}_1^f/\norm{\vec{d}_1^\circ}^2\right)$.

    As in the proof of Proposition~\ref{prop:make_collinear}, we repeatedly apply the first-order control $\vec{u}_{\Delta t}^{\epsilon, [\vec{h}_1,\vec{h}_2]}$ but now $r_1 = r_2$ so the particles rotate by equal angles and remain collinear.
\end{proof}

\begin{proposition}[Translating a group of equidistant collinear particles]\label{prop:translatingparticles}
Let the particles be initially equidistant and collinear, i.e., the initial configuration  $(\vec{x}^\circ,\vec{y}_1^\circ,\vec{y}_2^\circ)$ satisfies $\vec{d}_1^\circ = -\vec{d}_2^\circ$.
Given any scalar $\ell>0$ and unit vector $\vec{e}_{\perp} \perp \vec{d}_1^\circ$, the active and two passive particles can be translated by the same vector $\vec{\Delta} = \ell\vec{e}_{\perp}$ to the final configuration $(\vec{x}^f,\vec{y}_1^f,\vec{y}_2^f) = (\vec{x}^\circ+\vec{\Delta},\vec{y}_1^\circ+\vec{\Delta},\vec{y}_2^\circ+\vec{\Delta})$.
\end{proposition}
\begin{proof}
We consider a right-handed orthonormal reference frame $\{\vec{e}_1,\vec{e}_2,\vec{e}_3\}$ in which $\vec{e}_1 = \vec{e}_\perp$ and $\vec{e_2} = \vec{d}_2^\circ/\norm{\vec{d}_2^\circ}$. In this reference frame, $\vec{d}_1^\circ = (0, -d^\circ, 0)^\top$ and $\vec{d}_2^\circ = (0, d^\circ, 0)^\top$, where $d^\circ = \norm{\vec{d}_1^\circ}$.

Our strategy is to use the zeroth-order control to move the particles along the $\vec{e}_1$ direction. Since the passive particles move less than the active particle, they gradually fall behind. We use the first-order control to bring the passive particles in front of the active particle as needed and use symmetry to arrive with the same relative configuration as the initial state [see Figure~\ref{fig:compound}(e)].
    
More specifically, the first stage of our strategy involves applying the zeroth-order control $\vec{u}_{\Delta t}^{\ell/2, \vec{h}_1}$ moves the active particle to $\bar{\vec{x}} = \vec{x}^\circ + \tfrac{\ell}{2}\vec{e}_1$. For small $\ell$, the displacements of the passive particles are approximated by \eqref{formula11}. To leading order in $\ell$, the two passive particles undergo the same displacement, which is in the~$\vec{e}_1$ direction. For finite (possibly large) $\ell$, the exact displacements of the two passive particles are constrained by symmetry to have the form $\vec{\Delta}_{\vec{y}_1} = ({\Delta}_1, \Delta_2, 0)^\top$ and $\vec{\Delta}_{\vec{y}_2} = (\Delta_1, -\Delta_2, 0)^\top$. Hence, the relative position vectors of the passive particles in this configuration are of the form $\bar{\vec{d}}_1=\bar{\vec{y}}_1-\bar{\vec{x}}=(-\bar{d}_{1},-\bar{d}_{2}, 0)^\top$ and  $\bar{\vec{d}}_2=\bar{\vec{y}}_2-\bar{\vec{x}}=(-\bar{d}_{1}, \bar{d}_{2}, 0)^\top$. Since the component of the velocity in the $\vec{e}_1$ direction is always larger for the active particle than for the passive particles, we expect $\bar{d}_1 > 0$. We note, however, that this observation is not necessary for our proof.
    
In the second stage, we use Proposition~\ref{prop:rotatingparticles} to rotate the passive particles by the angle $\pi$ about the $\vec{e}_2$ axis through $\bar{\vec{x}}$ to achieve the configuration $(\bar{\vec{x}}, \tilde{\vec{y}}_1,\tilde{\vec{y}}_2)$ with relative position vectors $\tilde{\vec{d}}_1 = (\bar{d}_1, -\bar{d}_{2}, 0)^\top$ and $\tilde{\vec{d}}_2 = (\bar{d}_1, \bar{d}_2, 0)^\top$.

In the third stage, We apply the zeroth-order control $\vec{u}_{\Delta t}^{\ell/2, \vec{h}_1}$, bringing the active particle to $\vec{x}^f = \bar{x} + \tfrac{\ell}{2}\vec{e}_1 = \vec{x}^\circ + \ell\vec{e}_\perp$. This is equivalent to a time-reversal of the control applied at the beginning of our strategy in a coordinate frame that has been rotated by $\pi$ about the $\vec{e}_2$ axis through $\bar{\vec{x}}$. Hence, the displacements of the passive particles are the negative of the displacements $\vec{\Delta}_{\vec{y}_1}$ and $\vec{\Delta}_{\vec{y}_2}$ described earlier, rotated about the $\vec{e}_2$ direction. The final positions of the passive particles are, therefore, $\vec{y}_1^f = \vec{x}^f + \vec{d}_1^\circ = \vec{y}_1^\circ + \ell\vec{e}_\perp$ and $\vec{y}_2^f = \vec{x}^f + \vec{d}_2^\circ = \vec{y}_2^\circ + \ell\vec{e}_\perp$.

The procedure described above achieves the desired outcome provided that the configuration remains well separated at all times. Note that the second stage of the strategy does not alter distances between particles and the third stage is a rotated reversal of the first stage. Hence, we need only consider the changes in distances during the first stage of our strategy. In this stage, the distances $r_1=\norm{\vec{y}_1 - \vec{x}}$ and $r_2=\norm{\vec{y}_2-\vec{x}}$ increase from their initial values $r_1^\circ$ and $r_2^\circ$, respectively, because the active particle moves faster than the passive particles in a direction away from them. On the contrary, the distance between the passive particles, $\norm{\vec{y}_2-\vec{y}_1}=2\bar{d}_2$, decreases during this stage because the passive particles move towards each other in the $\vec{e}_2$-direction and always have the same position in the $\vec{e}_1$-direction. Hence, it is possible for the distance between the passive particles to decrease to the well-separated limit $R$. 

Suppose that this would occur at the point where the active particle has traveled a distance $\ell^*$. To avoid reaching this point, we break up the motion into a number $N_\ell$ of shorter segments of length $\bar{\ell}=\ell/N_\ell < \ell^*$, so that we may accomplish translations of the particles by a displacement vector $\bar{\ell}\vec{e}_1$ without violating the well-separated condition. The repetition of this shortened motion achieves the desired final outcome and maintains the desired separation.
\end{proof}

\begin{proposition}[Adjusting distances between particles in an equidistant collinear configuration]\label{prop:adjust_distances}
The distance between the active and passive particles in an equidistant collinear configuration can be changed arbitrarily. That is, given an initial configuration $(\vec{x}^\circ,\vec{y}_1^\circ = \vec{x}^\circ -r^\circ\vec{e}_2, \vec{y}_2^\circ = \vec{x}^\circ + r^\circ \vec{e}_2)$ with $r^\circ > R$, there is a control that achieves the final configuration $(\vec{x}^f=\vec{x}^\circ,\vec{y}_1^f = \vec{x}^\circ -r^f\vec{e}_2, \vec{y}_2^f = \vec{x}^\circ + r^f \vec{e}_2)$ with $r^f > R$. 
\end{proposition}
\begin{proof}
We first describe how to achieve a final distance $r^f < r^\circ$. Repeated applications of the second-order control $\vec{u}_{\Delta t}^{\epsilon, [\vec{h}_2,[\vec{h}_1,\vec{h}_2]]}$ leave the active particle at the initial position and move the passive particles along curves in the $\vec{e}_1$--$\vec{e}_2$ plane that lead to the active particle, shown as streamlines in Figure~\ref{fig:vectorfields}(c). By symmetry, the passive particles maintain equal and opposite displacements in the $\vec{e}_2$-direction. We stop applying this control when we reach relative positions $\vec{e}_2\cdot\vec{d}_2 = -\vec{e}_2\cdot \vec{d}_1 = r^f$.

We then apply the second-order control $\vec{u}_{\Delta t}^{\epsilon, [\vec{h}_3,[\vec{h}_1,\vec{h}_3]]}$. By \eqref{translation_displacement2}, replacing the index 2 with 3, this control moves the two passive particles in the $\vec{e}_1$-direction. We repeat this control until we achieve $\vec{e}_1\cdot \vec{d}_1 = \vec{e}_1\cdot \vec{d}_2 = 0$, which then satisfies the desired final configuration. The complete process is illustrated in Figure~\ref{fig:compound}(f).

By the assumption that $r^f > R$, we have that $\norm{\vec{d}_j} \geq \vec{e}_2\cdot \vec{d}_j \geq r^f >R$ for $j=1,2$ and $\norm{\vec{y}_2-\vec{y}_1} = 2\norm{\vec{d}_1} > R$ throughout this strategy so the configuration remains well separated.

In the case that $r^f > r^\circ$, we apply the control strategy above in reverse.
\end{proof}

\section{Controllability for one or two passive particles}
\label{sec:controllability_M_1_2}
In this section, we prove controllability results for systems of one active particle and one or two passive particles, based on system \eqref{MM05}. We first prove the case for two passive particles. Note that controllability with one passive particle follows from the controllability with two passive particles, since the two passive particles act as tracers and do not affect the dynamics of each other or of the active particle. The control strategy, however, can be simplified for a single passive particle so we will present a separate proof.

\begin{theorem}[Controllability with $M=2$ passive particles]\label{222}
An active particle and two passive particles can be moved from any well-separated initial configuration $(\vec{x}^\circ, \vec{y}_1^\circ,\vec{y}_2^\circ)$ to any well-separated final configuration $(\vec{x}_1^f, \vec{y}_1^f, \vec{y}_2^f)$ along a well-separated trajectory. That is, given $(\vec{x}^\circ, \vec{y}_1^\circ,\vec{y}_2^\circ), (\vec{x}_1^f, \vec{y}_1^f, \vec{y}_2^f) \in \mathcal{S}_R^3$\,, there exist
$T\in(0,+\infty)$ and a control map $\vec{u}\in L^\infty([0,T];\mathbb{R}^3)$ such that system \eqref{MM05} with $M=2$, namely,
\begin{equation*}%\label{MM06-11}
\begin{cases} 
\dot{\vec{x}}(t)=\vec{u}(t), \\[2mm]
\displaystyle \dot{\vec{y}}_1(t)=  \frac{3a}{4}\left(\frac{1}{\norm{\vec{d}_1(t)}}\One+\frac{1}{\norm{\vec{d}_1(t)}^3}\vec{d}_1(t)\otimes \vec{d}_1(t)\right){\vec{u}}(t),\\[2mm]
\displaystyle \dot{\vec{y}}_2(t)=  \frac{3a}{4}\left(\frac{1}{\norm{\vec{d}_2(t)}}\One+\frac{1}{\norm{\vec{d}_2(t)}^3}\vec{d}_2(t)\otimes \vec{d}_2(t)\right){\vec{u}}(t),
\end{cases}
\end{equation*}
admits a unique solution $(\vec{x},\vec{y}_1,\vec{y}_2)\in AC([0,T];\mathcal{S}_R^{3})$ (depending on $\vec{u}$), such that $(\vec{x}(0),\vec{y}_1(0),\vec{y}_2(0))=(\vec{x}^\circ,\vec{y}_1^\circ, \vec{y}_2^\circ)$ and $(\vec{x}(T),\vec{y}_1(T),\vec{y}_2(T))=(\vec{x}^f,\vec{y}_1^f,\vec{y}_2^f)$.
\end{theorem}

\begin{proof} We first describe three parts of the control strategy and then explain how they are combined.
    
    \noindent
    \emph{Part 1.}
    Using Propositions~\ref{prop:make_collinear} and~\ref{prop:make_equidistant}, we can bring the particles from the arbitrary initial configuration $(\vec{x}^\circ, \vec{y}_1^\circ,\vec{y}_2^\circ)$ to an intermediate configuration $(\tilde{\vec{x}}^\circ, \tilde{\vec{y}}_1^\circ,\tilde{\vec{y}}_2^\circ)$ that is equidistant and collinear. 

    \noindent
    \emph{Part 2.}
    Likewise, we can bring the particles from the final configuration $(\vec{x}_1^f, \vec{y}_1^f, \vec{y}_2^f)$ to a configuration $(\tilde{\vec{x}}^f, \tilde{\vec{y}}_1^f,\tilde{\vec{y}}_2^f)$ that is equidistant and collinear. 

    \noindent
    \emph{Part 3.}
    Using Propositions~\ref{prop:rotatingparticles},~\ref{prop:translatingparticles}, and~\ref{prop:adjust_distances}, the particles can be moved from $(\tilde{\vec{x}}^\circ, \tilde{\vec{y}}_1^\circ,\tilde{\vec{y}}_2^\circ)$ to $(\tilde{\vec{x}}^f, \tilde{\vec{y}}_1^f,\tilde{\vec{y}}_2^f)$.
    
    The sequence of steps in our complete strategy is, therefore: (i)  apply Part 1; (ii) apply Part 3; (iii) apply Part 2 in reverse. Indeed, by the time-reversibility property of the Stokes equation, we can reverse the control for Part 2 to bring particles from $(\tilde{\vec{x}}^f, \tilde{\vec{y}}_1^f,\tilde{\vec{y}}_2^f)$ to the given final configuration. This strategy takes inspiration from~\citep{dalmaso}.
    
    Each of the Parts 1--3 above can be achieved in finite time. Defining $T$ as the sum of these times, a control map $\vec{u}\in L^\infty([0,T];\mathbb{R}^3)$ can be constructed by concatenating the controls associated with the parts above. This control map steers the system from the initial conditions to the final conditions along a well-separated trajectory.
    The theorem is proved.
\end{proof}

\begin{theorem}[Controllability with $M=1$ passive particle]\label{221}
An active particle and a single passive particle can be moved from any well-separated initial configuration $(\vec{x}^\circ, \vec{y}^\circ)$ to any well-separated final configuration $(\vec{x}^f, \vec{y}^f)$ along a well-separated trajectory. That is, given $(\vec{x}^\circ, \vec{y}^\circ), (\vec{x}^f, \vec{y}^f) \in\mathcal{S}_R^2$\,, there exist
$T\in(0,+\infty)$ and a control map $\vec{u}\in L^\infty([0,T];\mathbb{R}^3)$ such that system \eqref{MM05} with $M=1$, namely,
\begin{equation*}%\label{MM06-11}
\begin{cases} 
\dot{\vec{x}}(t)=\vec{u}(t), \\[2mm]
\displaystyle \dot{\vec{y}}(t)=  \frac{3a}{4}\left(\frac{1}{\norm{\vec{d}(t)}}\One+\frac{1}{\norm{\vec{d}(t)}^3}\vec{d}(t)\otimes \vec{d}(t)\right){\vec{u}}(t),
\end{cases}
\end{equation*}
admits a unique solution $(\vec{x},\vec{y})\in AC([0,T];\mathcal{S}_R^{2})$ (depending on $\vec{u}$), such that $(\vec{x}(0),\vec{y}(0))=(\vec{x}^\circ,\vec{y}^\circ)$ and $(\vec{x}(T),\vec{y}(T))=(\vec{x}^f,\vec{y}^f)$.

\end{theorem}
\begin{proof}
We provide a constructive proof of controllability, which is achieved by an appropriate composition of the zeroth-, first-, and second-order controls. We apply the propositions from section \ref{sec:compound_moves}, which concerned systems with two passive particles; by simply neglecting the second passive particle, those propositions describe possible moves for an active particle and a single passive particle. 

Let us denote by $\Pi^\circ$ the plane containing $\vec{x}^\circ$,~$\vec{y}^\circ$, and~$\vec{x}^f$, and let us notice that it is not restrictive to assume that $\vec{x}^\circ$ is the origin. We may also choose the reference frame such that $\vec{e}_3$ is perpendicular to $\Pi^\circ$ and  $\vec{x}^f=x^f\vec{e}_1$ lies on the positive $x$-axis (see Figure~\ref{fig:one_active_one_passive}).

\smallskip

\noindent\emph{Step 1: rotation about $\vec{x}^\circ$}. 
By Proposition~\ref{prop:rotatingparticles}, the passive particle can be rotated about $\vec{x}^\circ$ to lie on the negative $\vec{e}_1$ axis (see Figure~\ref{fig:one_active_one_passive}(a)); its position at the end of this step will be $\hat{\vec{y}}=-d^\circ\vec{e}_1$, where $d^\circ=\norm{\vec{y}^\circ}$.

\smallskip

\noindent\emph{Step 2: translation of active particle}. Using the direct control $\vec{u}(t)=\vec{u}^{x^f,\, \vec{h}_1}_1(t) = \vec{x}^f-\vec{x}^\circ$, we move the active particle from~$\vec{x}^\circ$ to~$\vec{x}^f$. The passive particle moves in the $\vec{e}_1$ direction to $\bar{\vec{y}}$. Since the velocity of the active particle is greater than that of the passive particle, the particles remain well separated during this step.

\smallskip

\noindent\emph{Step 3: rotation about $\vec{x}^f$}. 
We consider the plane~$\Pi^f$ containing~$\bar{\vec{y}}$,~$\vec{x}^f$, and~$\vec{y}^f$, and change the reference frame, using orthonormal vectors $\vec{e}_1^f = \vec{d}^f/\norm{\vec{d}^f}$ and $\vec{e}_2^f$ in this plane. By Proposition~\ref{prop:rotatingparticles}, we can rotate the passive particle around the active one until it reaches a position $\bar{\vec{y}}^f=\bar{d}\vec{e}_1^f$ on the positive $x$-axis analogously to Step 1. 

\smallskip

\noindent\emph{Step 4: translation of passive particle (distance adjustment)}. Note that $\vec{y}^f$ and $\bar{\vec{y}}^f$ both lie on the positive $\vec{e}_1$-axis relative to the active particle $\vec{x}^f$. By Proposition~\ref{prop:adjust_distances}, we can adjust the distance between the active and passive particle to achieve the desired final configuration $(\vec{x}^f, \vec{y}^f)$. 

\smallskip

\noindent\emph{Conclusion of the proof}. Each of the steps 1--4 above can be achieved in finite time. Defining $T$ as the sum of these times, a control map $\vec{u}\in L^\infty([0,T];\mathbb{R}^3)$ can be constructed by concatenating the controls associated with the steps above. This control map steers the system from the initial conditions to the final conditions along a well-separated trajectory.
\begin{figure}
    \centering
    \includegraphics[width=0.75\textwidth]{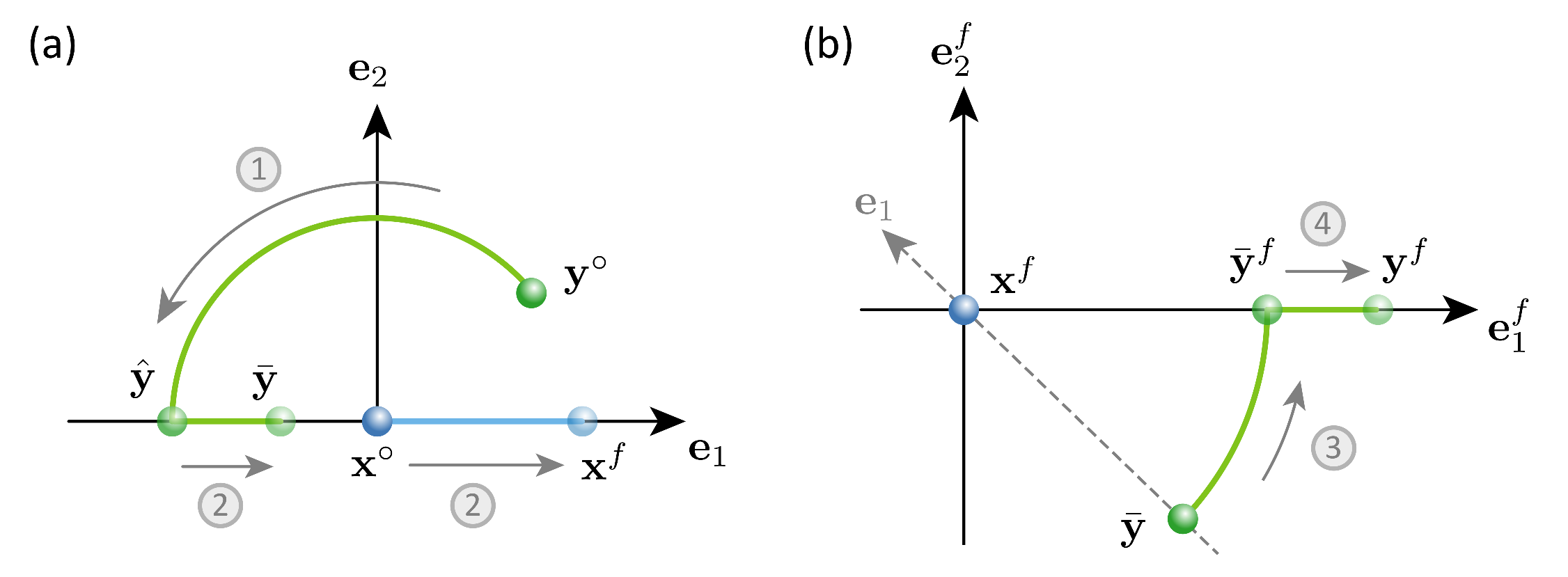}
    \caption{A schematic illustration of the four steps for moving one active and one passive particle from arbitrary initial positions $\vec{x}^\circ$, $\vec{y}^\circ$ to arbitrary final positions $\vec{x}^f$, $\vec{y}^f$. Steps 1 and 2 are shown in (a) while steps 3 and 4 are shown in (b) with a change of reference frame.
    \label{fig:one_active_one_passive}}
\end{figure}
\end{proof}

\section{Errors due to finite amplitudes and separations}\label{sec:errors}
In the control strategies for compound moves and the general controllability theorems of Section~\ref{sec:controllability_M_1_2}, we used the far-field hydrodynamic flow field associated with a moving particle and we used Lie brackets to generate the necessary directions of motion of the passive particles in the asymptotic limit $\varepsilon \to 0$. Since the displacement per cycle decreases as~$\varepsilon$ decreases, it may be preferable in practice to use a relatively large value of~$\varepsilon$. In this section, we present numerical results for solutions of system~\eqref{MM05}, applying the first- and second-order controls to achieve a fixed target angular or linear displacement of a single passive particle with various values of~$\varepsilon$. We characterize the error between the intended exact (target) displacement and the numerically computed displacement with finite~$\varepsilon$. Numerical trajectories were obtained using the \texttt{solve\_ivp} function with the RK45 ODE solver from the SciPy Python library. Additionally, we characterize the differences in displacements using the far-field approximation~\eqref{MM05} compared with applying the same controls to system~\eqref{MM01}, which includes the $O(1/r^3)$ potential dipole in the velocity field, represented by the $\nabla^2\mathcal{G}$ term in equation~\eqref{MM02}. We intentionally consider an initial separation that is only a few times the particle diameter and in the following two subsections, we show that the leading-order expressions based on far-field hydrodynamics from Section~\ref{subsec:Lie_brackets} give good estimates for the particle displacements; we expect that errors would be reduced if particles are further apart. 

\subsection{Angular displacements}
To characterize the errors associated with finite $\varepsilon$ when rotating a passive particle about the active particle, we consider a target angular displacement of $\theta=\pi/6$ about the $z$-axis for a passive particle initially at position $\vec{y}^\circ = (5, 0, 0)^\top$ and an active particle of radius $a=1$ at the origin. The target position for the passive particle is, accordingly, $\vec{y}_\mathrm{target} = (5\sqrt{3}/2, 5/2, 0)^\top$.

For a range of choices of integers $N$, we use formula~\eqref{Deltatheta} to define corresponding values of $\varepsilon$ such that each of $N$ applications of the first-order control $\vec{u}_{\Delta t}^{\varepsilon , \, [\vec{h}_1, \vec{h}_2]}$ is expected to produce a rotation by the angle $\Delta\theta^\varepsilon = \theta/N$ for the relevant values $a=1, r=5$. We then numerically solve system~\eqref{MM05} for $N$ applications of the control, obtaining the final position $\vec{y}^\varepsilon_\mathrm{num}$. 

Equation~\eqref{Deltatheta} neglects terms of order $\varepsilon^3$ from equation~\eqref{eq:first-order_displacement}. Hence, we can expect that the displacement deviates from the desired motion. By symmetry, the exact displacement of the passive particle has zero component in the $z$-direction, up to machine precision, even considering higher-order terms. The two components of interest are the error in the angular (polar) displacement, which in our case is
\begin{equation*}
    \eta_\theta^\varepsilon = \frac{1}{\theta} \arccos\left(\displaystyle\frac{\vec{y}^\varepsilon_\mathrm{num}}{\norm{\vec{y}^\varepsilon_\mathrm{num}}}\cdot\frac{\vec{y}_\mathrm{target}}{\norm{\vec{y}_\mathrm{target}}}\right),
\end{equation*}
and the radial error, which we define as 
\begin{equation*}
    \eta_r^\varepsilon = \frac{\left| \norm{\vec{y}^\varepsilon_\mathrm{num}}-r^\circ \right|}{r^\circ}.
\end{equation*}

The convergence of these two components of error is shown in Figure~\ref{fig:first_order_convergence} (labelled ``one-cycle polar'' and ``one-cycle radial'' in the legend). Both components of error appear to converge to zero with first-order rate of convergence with $\varepsilon$, which we could anticipate since equation~\eqref{Deltatheta} neglects terms of order $\varepsilon^3$ and the number of required iterations $N$ scales as $1/\varepsilon^2$.

We can, however, improve the rate of convergence by modifying the first-order control. The general first-order control~\eqref{eq:squarecontrol} applied with any of the following $(k,l)$ combinations of indices---$(1,2), (-1, -2), (2,-1), (-2,1)$---all result in the same leading-order term in the displacement~\eqref{eq:first-order_displacement}, but the higher-order terms differ. Here, a negative index indicates that we apply the control in the negative direction. As shown in Figure~~\ref{fig:first_order_convergence}, we find that a strategy of alternating between the first two choices, which we refer to as a two-cycle control, produces errors that decay quadratically with $\varepsilon$. A four-cycle strategy, in which we cycle through all four of the listed pairs of $(k,l)$, results in radial errors that decay cubically and polar errors that decay quadratically with $\varepsilon$. Note that values of $\varepsilon$ as large as 1 can be used for rotations with errors of $10^{-3}$ or less. The paths of the active and passive particles over one application of the four-cycle control are illustrated in Figure~\ref{fig:1_2_four_cycle}.

When we include the potential dipole term in the flow field, we find that the radial error is indistinguishable from the case without the potential dipole. In contrast, the polar error does not converge to zero with the four-cycle strategy but it remains below $10^{-2}$ for $\varepsilon < 1$.

We assert that it is more important to reduce the radial error than it is to guarantee a small polar error because we can easily adjust $\varepsilon$ or the number of iterations to compensate for errors in the angular component, whereas we require a different type of control to correct for changes in radial distance. In particular, the second-order control could generate a corrective radial displacement for one passive particle but this may not be able to correct the trajectories of two passive particles simultaneously.

Note that we considered particles that are initially close together, with $\norm{\vec{y}^\circ - \vec{x}^\circ} = 5a$. This was chosen to illustrate that even when the far-field regime is not strictly observed, the first-order control is an effective strategy for achieving circular motion of a passive particle around an active one. Radial deviations are small and angular displacements are well approximated by the leading order expression given by equation~\eqref{Deltatheta}. The analytical formula could be modified to include the effect of the potential dipole if greater accuracy were required.

\begin{figure}
    \centering
    \includegraphics[width=0.75\textwidth]{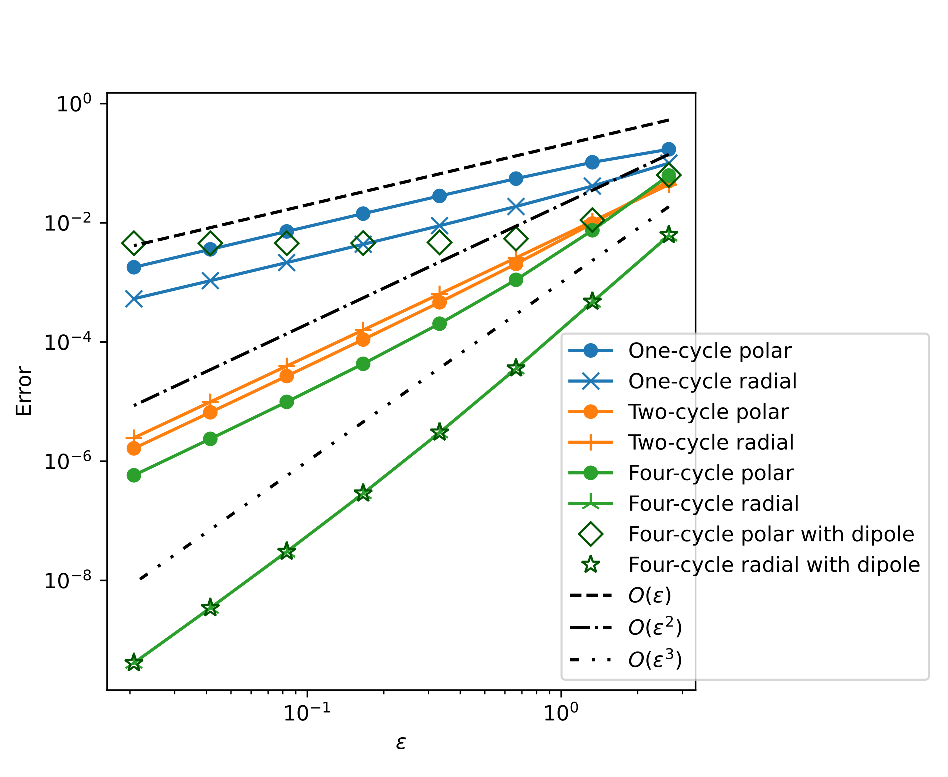}
    \caption{Convergence of the passive particle displacement with the number of applications of the first-order control corresponding to $[\vec{h}_1,\vec{h}_2]$ for a fixed target rotation by the angle $\pi/6$ about the origin.
    \label{fig:first_order_convergence}}
\end{figure}

\begin{figure}
    \centering
    \includegraphics[width=0.75\textwidth]{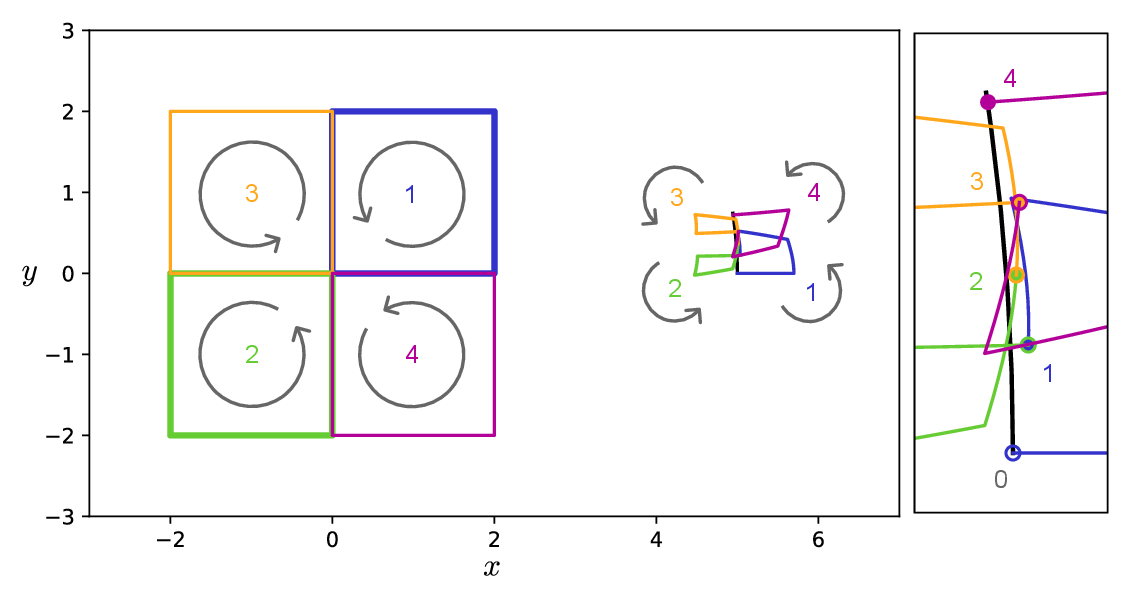}
    \caption{A multi-cycle control corresponding to the first-order Lie bracket $[\vec{h}_1,\vec{h}_2]$. Initially, the active particle is at $\vec{x}^\circ = (0,0,0)^\top$ and the passive particle is at $\vec{y}^\circ = (5,0,0)^\top$. The panel on the right is a magnification of the path of the passive particle. The two-cycle consists of portions 1 and 2, whereas the four-cycle consists of portions 1 to 4, in sequence. The amplitude of the control is $\varepsilon=2$.
    \label{fig:1_2_four_cycle}}
\end{figure}

\subsection{Linear displacements}
Errors for the second-order control are analyzed similarly, using a target displacement from $\vec{y}^\circ = (5,0,0)^\top$ to $\vec{y}_\mathrm{target}=(5.1,0,0)^\top$. For a given number $N$ of applications of the second-order control, we numerically determine the value of $\varepsilon$ that would result in the target displacement according to equation~\eqref{translation_displacement}, noting that $r$ in this equation changes with each application of the control. The polar component of error is defined as 
\begin{equation*}
    \eta_\theta^\varepsilon = \arccos\left(\frac{y^\varepsilon_{\mathrm{num}, 1}}{\norm{\vec{y}^\varepsilon_\mathrm{num}}}\right),
\end{equation*}
and the radial component of error for this motion is
\begin{equation*}
    \eta_r^\varepsilon = \frac{\left| y^\varepsilon_{\mathrm{num},1}-y_{\mathrm{target},1} \right|}{\left|y_{\mathrm{target},1}-y^\circ_1\right|} = 10\left| y^\varepsilon_{\mathrm{num},1}-y_{\mathrm{target},1} \right|.
\end{equation*}

Following the description in Section~\ref{sec:translation_control}, the one-cycle second-order control is obtained by applying the general pattern~\eqref{eq:squarecontrol} with $\vec{h}_k$ corresponding to the first-order control $[\vec{h}_1,\vec{h}_2]$ and $\vec{h}_l = -\vec{h}_2$. The two-cycle control alternates between this and the control with $\vec{h}_k$ corresponding to the first-order control $[\vec{h}_1,-\vec{h}_2]$ and $\vec{h}_l = \vec{h}_2$. The motion of the two particles generated by the two-cycle control is illustrated in Figure~\ref{fig:1_2_-2_two_cycle}.

As shown in Figure~\ref{fig:second_order_convergence}, the errors in the radial (linear) direction are similar for the one- and two-cycle controls over the range of $\varepsilon$ considered, decaying approximately linearly with $\varepsilon$. Polar errors decay quadratically with $\varepsilon$ with the one-cycle strategy and quintically for the two-cycle strategy. Both polar and radial errors are essentially unchanged when the potential dipole terms are included, as shown in Figure~\ref{fig:second_order_convergence}.

Since polar errors decay rapidly as $\varepsilon$ decreases, we can readily achieve linear motion of a passive particle using the second-order control. The errors in the radial direction, which are larger in magnitude than those in the polar direction, can be corrected either by considering higher-order terms in the analytic expression for the displacement or by reducing $\varepsilon$, perhaps incrementally as the target position is approached.

\begin{figure}
    \centering
    \includegraphics[width=0.75\textwidth]{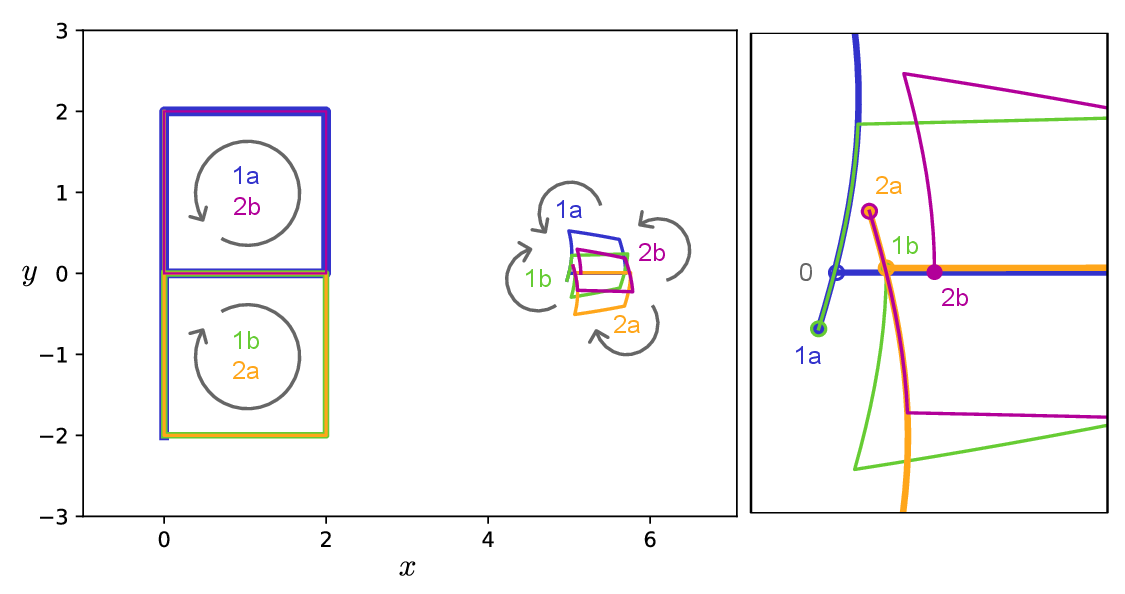}
    \caption{A multi-cycle control corresponding to the second-order Lie bracket $[[\vec{h}_1,\vec{h}_2],-\vec{h}_2]=[\vec{h}_2,[\vec{h}_1,\vec{h}_2]]$. Initially, the active particle is at $\vec{x}^\circ = (0,0,0)^\top$ and the passive particle is at $\vec{y}^\circ = (5,0,0)^\top$. The panel on the right is a magnification of the path of the passive particle. The direction of travel along the curves is indicated by arrows and the portions of the cycles correspond to applying the controls: (1a) $\vec{u}^{\epsilon,\, [\vec{h}_1,\vec{h}_2]}$, (1b) $\vec{u}^{\epsilon,\, -\vec{h}_2}$, (2a) $\vec{u}^{\epsilon,\, [\vec{h}_1,-\vec{h}_2]}$, and (2b) $\vec{u}^{\epsilon,\, \vec{h}_2}$ all with $\varepsilon=2$.
    \label{fig:1_2_-2_two_cycle}}
\end{figure}

\begin{figure}
    \centering
    \includegraphics[width=0.75\textwidth]{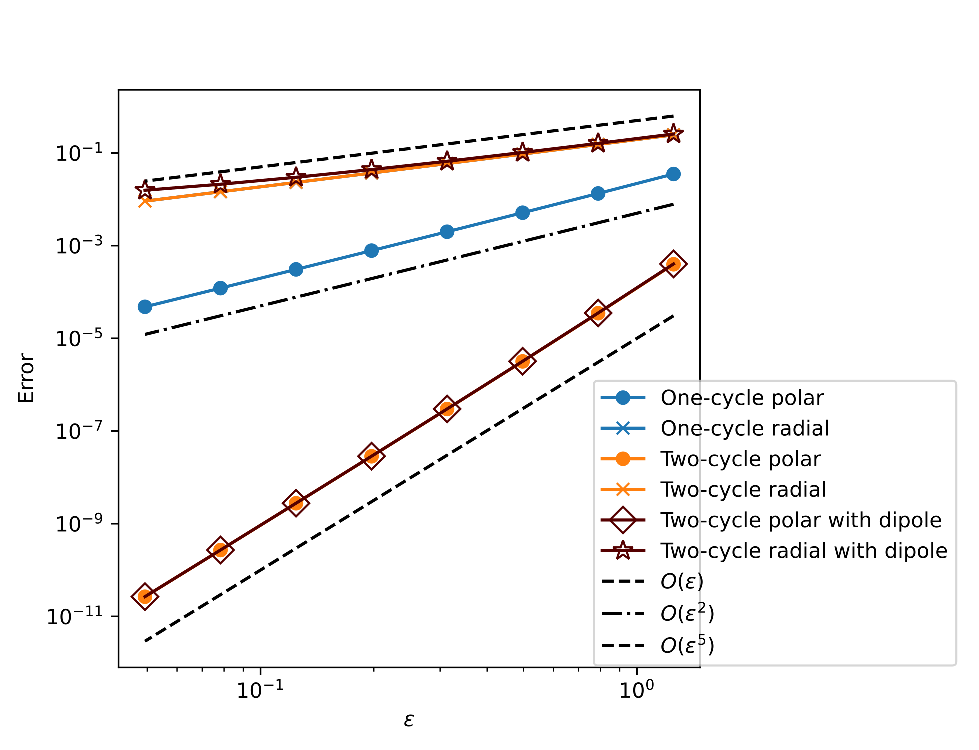}
    \caption{Convergence of the passive particle displacement with the control amplitude $\varepsilon$ for repeated applications of second-order controls corresponding to $[\vec{h}_2,[\vec{h}_1,\vec{h}_2]]$ for a fixed target displacement from $\vec{y}^\circ=(5,0,0)^\top$ to $\vec{y}_\mathrm{target} = (5.1,0,0)^\top$.
    \label{fig:second_order_convergence}}
\end{figure}

\section{Conclusions}\label{sec:conclusions}

In this paper, we have presented the motion planning problem for a system of one active and~$M$ passive spherical particles immersed in a viscous fluid, in the far-field approximation. We explicitly constructed elementary moves that, suitably concatenated, resulted in strategies to achieve total controllability in the specific cases $M=1$ and $M=2$. Moreover, the strategies we proposed ensure that the particles can maintain an arbitrary minimum separation compatible with their initial and final configurations. The elementary and compound moves were expressed in terms of zeroth-, first-, and second-order controls characterized by an amplitude parameter~$\varepsilon$, with asymptotic expressions valid in the limit $\varepsilon \to 0$. We showed that in this limit, the passive particle displacements resulting from the zeroth-, first-, and second-order controls correspond to the Stokeslet, rotlet, and rotlet dipole singularity solutions of Stokes flow, respectively. Higher-order singularities can be generated by extension of the controls to higher orders.

Through numerical solutions of the equations of motion, we demonstrated that the two key components of our motion planning strategy, namely, moving passive particles in a circular orbit around an active one and translating a passive particle without a net displacement of the active particle, could be achieved to a high accuracy even with $\varepsilon \approx 1$ and with the particles as close as a few diameters apart.

This research contributes to the growing literature on ensembles of microparticles subject to hydrodynamic interactions in low Reynolds number flows, including the possibly chaotic behavior of sedimenting particles~\citep{Hocking_1964,Janosi}, mixing and transport in suspensions of microswimmers~\citep{lauga_hydrodynamics_2009, katija_viscosity-enhanced_2009, Pushkin2013}, idealized models of swimmers such as Purcell's scallop or three-link swimmers~\citep{Purcell1977}, and three linked spheres~\citep{NajafiGolestanian}. Mathematical treatments of control of model swimmers started with the seminal paper~\cite{shapere_geometry_1989} and have since been applied in many contexts, such as~\cite{ADSL,chambrion_generic_2012,dalmaso,chambrion_optimal_2019,loheac_controllability_2020,ZMBG,AZN}

The present contribution sets the basis for further investigations from multiple viewpoints. Four areas of future research that could be of interest to the mathematical, physical, and engineering communities are:
\begin{description}
    \item[(1) ] Using periodic controls for the active particles to produce flow fields that act as hydrodynamic traps~\citep{lutz_hydrodynamic_2006,chamolly_irreversible_2020}. Rather than moving a passive particle from one specific location to another, we may want to attract all nearby particles to a target and hold them there, possibly against other effects such as a background flow or gravity. 
    \item[(2) ] Considering cases in which we have~$N$ active particles and~$M$ passive ones, with both $N$ and $M$ large. Generalizing the formulation~\eqref{MM07} to arbitrary numbers of active and passive particles is relatively straightforward but the task of effectively controlling $M$ passive particles with a minimal number of active particles is challenging. It could also be of interest to investigate whether partial controllability results can be proved for an even lower number of controllers.
    We stress that, even in the case $N=1$ and $M=3$, the strategies proposed in our proofs would have to be substantially modified, since the presence of a third passive particle disrupts the symmetry that has been exploited in some of the moves. For example, the strategy used in Proposition~\ref{prop:translatingparticles} (translating a group of equidistant collinear particles) does not work even if the particles are not all required to be collinear, because we cannot guarantee that the symmetry is preserved for the third passive particle.
    \item[(3) ] Mixing fluids at low Reynolds number. This is known to be challenging in microfluidic devices~\citep{ward_mixing_2015}; one proposed technique involves using magnetic particles in rotating magnetic fields~\citep{Munaz2017}, which corresponds to our model of actively actuated particles but with applied torques and rotations of the active particles playing significant roles.
    \item[(4) ] Accounting for stochastic terms (i.e., Brownian motion) in the dynamics of the passive particles~\citep{Graham2018}. In the current work, we assumed that particles were large enough that Brownian motion could be neglected but this may not be valid if the particles are small (more precisely, when the P\'{e}clet number is small). Including random diffusion effects would be particularly interesting when the number of passive particles is very large (ideally, diverging to infinity), to the point that a description in terms of the distribution of particles would be preferable. In this context, it is customary to study the PDE arising for the limiting distribution, which, in this context, is expected to be a Fokker--Planck-type equation featuring a transport term coming from the action of the active particles, with the diffusion term resulting from the Brownian motion. While this is a very promising and fertile research field, it is beyond the scope of the present paper.
\end{description}

\bigskip

\backsection[Acknowledgements]{
The authors thank AnhadSingh Bagga for assistance with some of the figures.
MZ is a member of the \emph{Gruppo Nazionale di Fisica Matematica} of the \emph{Istituto Nazionale di Alta Matematica}. MM is a member of the \emph{Gruppo Nazionale per l'Analisi Matematica, la Probabilit\`{a} e le loro Applicazioni} of the \emph{Istituto Nazionale di Alta Matematica}.
This study was carried out within the GNAMPA2024 project \emph{Analisi asintotica di modelli evolutivi di interazione} (CUP E53C23001670001).}

\backsection[Funding]{
HS and MA acknowledge the support of the Natural Sciences and Engineering Research Council of Canada (NSERC), [funding reference number RGPIN-2018-04418].

\noindent Cette recherche a \'{e}t\'{e} financ\'{e}e par le Conseil de recherches en sciences naturelles et en g\'{e}nie du Canada (CRSNG), [num\'{e}ro de r\'{e}f\'{e}rence RGPIN-2018-04418].

\noindent Funding from the \emph{Mathematics for Industry 4.0} 2020F3NCPX PRIN2020 (MM and MZ) funded by the Italian MUR,  the  \emph{Geometric-Analytic Methods for PDEs and Applications} 2022SLTHCE (MM) and the \emph{Innovative multiscale approaches, possibly based on Fractional Calculus, for the effective constitutive modeling of cell mechanics, engineered tissues, and metamaterials in Biomedicine and related fields} P2022KHFNB (MZ) projects funded by the European Union -- Next Generation EU  within the PRIN 2022 PNRR program (D.D. 104 - 02/02/2022) is gratefully acknowledged. 
This manuscript reflects only the authors’ views and opinions and the Ministry cannot be considered responsible for them.}

\backsection[Declaration of interests]{
The authors report no conflict of interest.}

\backsection[Author ORCIDs]{
H.~Shum, https://orcid.org/0000-0002-5385-1568;
M.~Zoppello, https://orcid.org/0000-0001-6659-4268; 
M.~Astwood, https://orcid.org/0000-0002-8830-3852; 
M.~Morandotti, https://orcid.org/0000-0003-3528-6152.}

\bibliographystyle{jfm}
%\bibliography{bib}

\end{document}